\documentclass{article}
\usepackage{amssymb}

\usepackage{amsmath,amssymb,amsthm}

\newtheorem{problem}{Problem}
\newtheorem{lemma}{Lemma}[section]
\newtheorem{proposition}{Proposition}[section]

\begin{document}
\title{Sparse solution of overdetermined linear systems when the columns of $A$ are orthogonal}
\author{Phanindra V. Jampana \\
  Department of Chemical Engineering \\
  Indian Institute of Technology, Hyderabad 502205, India \\
  Email: pjampana@iith.ac.in \\\\
  Challa S. Sastry \\
  Department of Mathematics\\
  Indian Institute of Technology, Hyderabad 502205, India \\
  Email: csastry@iith.ac.in}
\maketitle
\begin{abstract} In this paper, we consider the problem of obtaining the best $k$-sparse solution of
  $Ax=y$ subject to the constraint that the columns of $A$ are orthogonal. The naive approach for 
  obtaining a solution to this problem has exponential complexity and there exist $l_1$ regularization
  methods such as Lasso to obtain approximate  solutions. In this paper, we show that we can obtain an exact
  solution to the problem, with much less computational effort compared to the brute force search when the columns of $A$ are orthogonal.
\end{abstract}
\section{Introduction}

\par We consider the following problem:
\begin{problem}[$P_1$]
  Find $u^k \in \mathbb{R}^n$ such that
  \[\|Au^k - y\| = \inf\{\| Ax - y\|~|~x \in \mathbb{R}^n, \| 
  x \|_0 = k \}.\]
\end{problem}

\noindent Here, the dimensions of $u^k,A,y,x$ are $n\times 1, m \times n, m \times 1, n \times 1$ respectively, with $m > n$. Obtaining sparse solutions to overdetermined system of equations has a long history in the statistics community. For example, the Lasso algorithm due to Tibshirani 
(\cite{tibshirani96regression}) tries to solve the following problem (for a fixed $\lambda$):-
\[\mathrm{~find~} x \mathrm{~such~that ~} \|Ax-y\|_2^2 + \lambda \|x\|_1 \mathrm{~is~minimized~} \] 

$l_1$ norm is employed here rather than the $l_0$ as using $l_1$ penalization makes the problem tractable
(convex optimization methods can be used to solve it) and also due to the fact that minimizing $l_1$ 
norm typically provides sparse solutions. The minimizer to the Lasso is obtained by solving a series of Quadratic Progrmmaing problems (\cite{tibshirani96regression}). Other computational techniques also exist in the literature for solving the Lasso, e.g. see \cite{osborne,lassodual}.

The Lasso method is general and is applicable to any matrix. However, the parameter $\lambda$ in the unconstrained formulation of the problem has to be tuned to obtain satisfactory results. $\lambda$ is usually
obtained by cross-validation.

Sparse solutions to overdetermined system are also considered in the paper by
\cite{Candes_Rudelson_Tao_Vershynin_2005}. The authors study the problem of reconstructing $x$
exactly when the observed data are corrupted by noise. If 
\[ y = Ax + e\]
the authors give conditions on the matrix $A$ and an minimizing algorithm which recovers $x$ \textit{exactly}
subject to a constraint on the number of non-zero entries of $e$. Howerver, this work is not 
directly relevant to the problem under consideration.

In this paper, we give an explicit solution to $P_1$, under the constraint that the columns of $A$ are orthogonal. We show that the solution given is equivalent (i.e. has equal error) to any solution obtained by a brute force search. One advantage of the method over the Lasso is that no tuning is necessary. As the
proposed method still involves computing inverse of $A^TA$, it might not scale well to problems where is
$n$ is very large.

\section{Equivalence of solutions}
We first fix some notation. If $x$ is a vector let $x_{\circ}^2$ denote the element wise square of $x$.
Let, $z = (A^TA)^{-1}(A^Ty)_{\circ}^2$ and $z_s$ be the result of sorting $z$ (in a stable manner) 
in the decreasing order. 
Let $f$ be a permutation such that $f(i) = j$ implies that $z_s(i) = z(j)$, where $x(i)$ denotes the
$i^\mathrm{th}$ element of the vector $x$. Define, $x_{\mathrm{pi}} = A^\dagger y$ and finally,
\begin{eqnarray*}
v^k(f^{-1}(i)) &=& x_{\mathrm{pi}} (f^{-1}(i)) \mathrm{~if~} i \leq k\\
v^k(f^{-1}(i)) &=& 0, \mathrm{~if~} i > k
\end{eqnarray*}
Note here that $v^k$ can be computed with much less computational effort than the brute force search as 
the inverse $(A^TA)^{-1}$ is only computed once. 

\begin{proposition}
  \label{prop:prop1} 
  Let $A$ be an $m \times n$ matrix such that $m > n$, $n > 1$. Assume that $A$ has full column rank. If 
  the columns of $A$ are orthogonal, then
  \[\|Au^k-y\| = \|Av^k-y\|~\forall y \in \mathbb{R}^m, ~\forall 
  k \in \{1,...,n\}\]
\end{proposition}
\begin{proof}
  As $u^k$ is a solution to P1, its non-zero elements should be of the form $u = (AC)^\dagger y$ 
  where $C$ is a $n\times k$ column picking matrix. Similarly, the non-zero elements of 
  $v^k$ can be written as $v = RA^\dagger y$ where $R$ is a $k \times n$ row picking matrix. 
  We then have $\|Au^k - y\| = \|ACu - y\|$ and 
  $\|Av^k-y\| = \|AR^Tv-y\|$. Therefore,
  \begin{equation}
  \label{eqn:simplific_1}
    \begin{split}
      (ACu-y)^T&(ACu-y) \\
      = &~u^TC^TA^TACu - 2y^TACu +y^Ty\\
      = &~y^T((AC)^\dagger)^TC^TA^TAC(AC)^\dagger y - 2y^TAC(AC)^\dagger y+y^Ty\\ 
      = & ~-y^TAC(AC)^\dagger y +y^Ty\\ 
    \end{split}
  \end{equation}
  
  \begin{equation}
  \label{eqn:simplific_2}
    \begin{split}
      (AR^Tv-y)^T&(AR^Tv-y) \\
      = &v^TRA^TAR^Tv-2y^TAR^Tv+y^Ty \\
      = &y^T(A^\dagger)^TR^TRA^TAR^TRA^\dagger y - 2y^TAR^TRA^\dagger y +y^Ty\\
    \end{split}
  \end{equation}

  \noindent Now, we prove that if the columns of $A$ are assumed to be orthogonal i.e. if $A^TA$ is a diagonal matrix
  then $(ACu-y)^T(ACu-y) = (AR^Tv-y)^T(AR^Tv-y)$. To show this, we first try to find the $C$ which 
  minimizes $(ACu-y)^T(ACu-y)$. As has been shown above, this is equivalent to minimizing
  \[-y^TAC(AC)^\dagger y + y^Ty\] or maximizing
  \begin{equation}
    \label{eqn:onf}
    (A^Ty)^T C(C^TA^TAC)^{-1}C^T (A^Ty)
  \end{equation}
  As $A^TA$ is diagonal, $C(C^TA^TAC)^{-1}C^T = CC^T(A^TA)^{-1}$ and hence,
  \begin{equation}
    \begin{split}
      (A^Ty)^T C(C^TA^TAC)^{-1}C^T (A^Ty) & = \sum_{i \in \{j | (CC^T)_{jj} \neq 0\}} \frac{(A^Ty)_i^2}{\lambda_i}\\ 
    \end{split}
  \end{equation}
  where, $\lambda_i$'s are the diagonal elements of $A^TA$
  (note that the $\lambda_i$'s are strictly positive real numbers).
  The maximum possible value of 
  $\sum_{i \in \{j | (CC^T)_{jj} \neq 0\}} \frac{(A^Ty)_i^2}{\lambda_i}$ is
  $\sum_{i \in \{j | (R^TR)_{jj} \neq 0\}} \frac{(A^Ty)_i^2}{\lambda_i}$ 
  as $R$ picks the maximum $k$-components of $(A^TA)^{-1}(A^Ty)_{\circ}^2$. So without loss
  of generality we can assume that $C=R^T$ as the error $\|ACu-y\|$ cannot
  be minimised any further.

  The $i^{\mathrm{th}}$ row of the matrix $R^TRA^TA$ equals the 
  $i^{\mathrm{th}}$ row of $A^TA$ if $R_{ii} = 1$ and equals the zero row 
  otherwise. Therefore, 
  \[(A^TA)^{-1}R^TR(A^TA)R^TR (A^TA)^{-1} - 2R^TR(A^TA)^{-1} = -R^TR(A^TA)^{-1}\]
  From $C=R^T$ we finally get
  \[(A^TA)^{-1}R^TR(A^TA)R^TR (A^TA)^{-1} +  C(C^T(A^TA)C)^{-1} C^T 
  - 2R^TR(A^TA)^{-1} = 0\]
\end{proof}

The above analysis raises the following question: can we say anything in the reverse?  Supposing we are 
given that $\|Av^k-y\| = \|Au^k-y\|$ for all $y$, then is it true that the columns 
of $A$ are orthogonal? We show below that this is indeed true for the case of $k=1$. For the proof of 
this fact we need the following supporting lemma.

\begin{lemma}
  \label{lem:inverdiag}
  Let $A$ be an $m \times n$ matrix such that $m > n$, $n > 1$. Assume that $A$ has full column rank.
  If the diagonal entries of $(A^TA)^{-1}$ are inverses of the diagonal entries of $A^TA$ (i.e.,
  if $(A^TA)^{-1}_{ii}(A^TA)_{ii} = 1$) then the off-diagonal elements of $A^TA$ (and hence the off-diagonal 
  elements of $(A^TA)^{-1}$) are all equal to zero.
\end{lemma}
\begin{proof}
  First the claim is proved for $n=2$ and the general case is proved by induction. For the case $n=2$ assume
  that $A^TA = \left( \begin{array}{cc} \alpha & \beta \\ \beta & \gamma \end{array} \right)$ and 
  $(A^TA)^{-1} = \left(\begin{array}{cc} \frac{1}{\alpha} & \beta^\prime \\ \beta^\prime & \frac{1}{\gamma} 
  \end{array} \right)$.
  From,
  \[\left( \begin{array}{cc} \alpha & \beta \\ \beta & \gamma \end{array} \right) * 
  \left(\begin{array}{cc} \frac{1}{\alpha} & \beta^\prime \\ \beta^\prime & \frac{1}{\gamma} \end{array} \right) =
  \left( \begin{array}{cc} 1 & 0 \\ 0 & 1 \end{array} \right)\]
  we get that $\beta\beta^\prime = 0$ and $\alpha\beta^\prime + \frac{\beta}{\gamma} = 0$, which implies that 
  $\beta = \beta^\prime = 0$ proving the claim for this case.
  Now, we assume that the proposition is true for $n-1$. Let
  \[A^TA= \left( \begin{array}{cc} \Sigma_{11} & \Sigma_{12} \\ \Sigma_{21} & \Sigma{22} \end{array} \right)\]
  where $\Sigma_{11}$ is an $n-1\times n-1$ matrix, $\Sigma_{22}$ is a scalar, $\Sigma_{12}$ is a $n-1\times 1$ vector
  and $\Sigma_{21}$ is a $1\times n-1$ vector. It is important to note here that $A^TA$ and $\Sigma_{11}$ are symmetric
  positive definite matrices. Therefore, their inverses $(A^TA)^{-1}$ and $\Sigma_{11}^{-1}$ are 
  also symmetric and positive definite (\cite{matmathper}, Corollary 14.2.11). Using blockwise matrix 
  inversion, we can write $(A^TA)^{-1}$ as (see \cite{matmath}, p. 45)
  \[\left( \begin{array}{cc} \Sigma_{11}^{-1} 
    + \Sigma_{11}^{-1}\Sigma_{12}(\Sigma_{22} -\Sigma_{21}\Sigma_{11}^{-1}\Sigma_{12})^{-1}\Sigma_{21}\Sigma_{11}^{-1} 
    & -\Sigma_{11}^{-1}\Sigma_{12}(\Sigma_{22} -\Sigma_{21}\Sigma_{11}^{-1}\Sigma_{12})^{-1}\\
    -(\Sigma_{22} -\Sigma_{21}\Sigma_{11}^{-1}\Sigma_{12})^{-1}\Sigma_{12}\Sigma_{11}^{-1}
    & (\Sigma_{22} -\Sigma_{21}\Sigma_{11}^{-1}\Sigma_{12})^{-1} \end{array} 
    \right)\] 
    Here, $(\Sigma_{22} -\Sigma_{21}\Sigma_{11}^{-1}\Sigma_{12})$ is 
    the schur complement of $\Sigma_{11}$ in $A^TA$ and is positive definite as 
    $A^TA$ and $\Sigma_{11}$ are both positive definite (see, \cite{convopt}, Appendix A.5.5). Hence, the 
    above blockwise matrix inversion formula is valid.
    Now, from
    \[(\Sigma_{22} -\Sigma_{21}\Sigma_{11}^{-1}\Sigma_{12})^{-1}*\Sigma_{22} = 1\]
    we get that $\Sigma_{21}\Sigma_{11}^{-1}\Sigma_{12} =\Sigma_{12}^T\Sigma_{11}^{-1}\Sigma_{12}=0$. 
    As $\Sigma_{11}^{-1}$ is 
    positive definite, we obtain that $\Sigma_{12} = 0$. Therefore,
    \[(A^TA)^{-1}=\left( \begin{array}{cc} \Sigma_{11}^{-1}  
    & 0\\
    0
    & \Sigma_{22}^{-1} \end{array} 
    \right)\] 
    From the induction hypothesis we know that $\Sigma_{11}^{-1}$ is a diagonal matrix. Therefore, 
    $(A^TA)^{-1}$ is diagonal and the induction step is proved.
\end{proof}

\begin{proposition}
  \label{prop:prop2} 
  Let $A$ be an $m \times n$ matrix such that $m > n$, $n > 1$. Assume that $A$ has full column rank. Then, 
  \[\|Au^1-y\| = \|Av^1-y\|~\forall y \in \mathbb{R}^m\] if and only if the columns of $A$ 
  are orthogonal.
\end{proposition}

\begin{proof}
  The forward implication has already been proved previously and 
  to prove the reverse implication, we show the existence of a few $y$'s so that if 
  expression if $\|Au^1-y\| = \|Av^1-y\|$ for all these choices of $y$ then $A^TA$ is diagonal.
  We first choose $y = A(A^TA)^{-1}[1,0,0,\cdots,0]^T$ and look for 
  solutions $u^1$ and $v^1$. 

  For the above choice of $y$ expression
  (\ref{eqn:onf}) reduces to:-
  \[[1,0,0,\cdots,0] C(C^TA^TAC)^{-1}C^T [1,0,0,\cdots,0]^T\]
  It is easy to see that the $C$ which 
  maximizes the above expression is $C = [1,0,0,\cdots,0]^T$ as for any other choice of $C$ 
  the expression equals zero. 

  Now, we shall show that  
  \begin{equation}
    \label{eqn:ong}
    \begin{split}
      (A^Ty)^T((A^TA)^{-1}R^TR(A^TA)R^TR (A^TA)^{-1} - 2R^TR(A^TA)^{-1} \\
      +  C(C^T(A^TA)C)^{-1} C^T)(A^Ty)
      \end{split}
  \end{equation}
  can equal zero only for the choice $R=[1,0,0,\cdots,0]$. For this, let
  \[A^TA = \left( \begin{array}{cccc} 
    p_{11} & p_{12} & \cdots & p_{1n} \\
    p_{21} & p_{22} & \cdots & p_{2n} \\
    \vdots & \vdots & \cdots & \vdots \\
    p_{n1} & p_{n2} & \cdots & p_{nn} \\
  \end{array} \right), (A^TA)^{-1} = \left( \begin{array}{cccc} 
    j_{11} & j_{12} & \cdots & j_{1n} \\
    j_{21} & j_{22} & \cdots & j_{2n} \\
    \vdots & \vdots & \cdots & \vdots \\
    j_{n1} & j_{n2} & \cdots & j_{nn} \\
  \end{array} \right) \]
  When $R = [1,0,0,\cdots,0]$ we get that the value of expression (\ref{eqn:ong}) is equal to 
  $p_{11}j_{11}^2 - 2j_{11} + \frac{1}{p_{11}}$. This value can be made zero by choosing 
  $j_{11} = \frac{1}{p_{11}}$. For any
  other choice of $R$ such that $R_i = 1, i \neq 1$, the value of expression (\ref{eqn:ong}) is equal to 
  $p_{ii}j_{1i}^2 + \frac{1}{p_{11}}$. This value cannot be made zero by any choice of $j_{1i}$ as 
  $p_{11} > 0$.
  As expression (\ref{eqn:ong}) has to equal zero from our initial assumption, we are forced to choose
  $j_{11} = \frac{1}{p_{11}}$.

  Now, we choose $y = A(A^TA)^{-1}[0,1,0,\cdots,0]^T$ and obtain that $j_{22} = \frac{1}{p_{22}}$. 
  By continuing in this fashion, we get that $j_{ii} = \frac{1}{p_{ii}}~\forall i \in \{1,2,3,\cdots,n\}$. 
  Finally, we apply 
  Lemma \ref{lem:inverdiag} and get that both $A^TA$ and $(A^TA)^{-1}$ are diagonal.
\end{proof}
\bibliographystyle{harvard}
\bibliography{sparse_od_ortho_jampana_challa}

\end{document}